\documentclass[aps,pra, twocolumn,superscriptaddress,nofootinbib]{revtex4-1}

\usepackage{amsfonts,amssymb,amsmath}
\usepackage[]{graphics,graphicx,epsfig}
\usepackage{amsthm}
\usepackage{color}
\usepackage[utf8]{inputenc}
\usepackage[ampersand]{easylist}
\usepackage{hyperref}

\newtheorem{theorem}{Theorem}
\newtheorem{proposition}{Proposition}

\newtheorem{lemma}{Lemma}

\theoremstyle{definition}

\newcommand{\ket}[1]{\left | #1 \right\rangle}
\newcommand{\bra}[1]{\left \langle #1 \right |}

\newcommand{\Tr}{\mathrm{Tr}}

\newcommand{\proj}[1]{\ket{#1}\bra{#1}}
\renewcommand{\epsilon}{\varepsilon}

\def\identity{\leavevmode\hbox{\small1\kern-3.8pt\normalsize1}}

\begin{document}

\title{Entanglement gain in measurements with unknown results}

\author{Margherita Zuppardo}
\affiliation{School of Physical and Mathematical Sciences, Nanyang Technological University, 637371 Singapore, Singapore}
\affiliation{Science Institute, University of Iceland, IS-107 Reykjavik, Iceland}

\author{Ray Ganardi}
\affiliation{School of Physical and Mathematical Sciences, Nanyang Technological University, 637371 Singapore, Singapore}

\author{Marek Miller}
\affiliation{School of Physical and Mathematical Sciences, Nanyang Technological University, 637371 Singapore, Singapore}

\author{Somshubhro Bandyopadhyay}
\affiliation{Department of Physics and Center for Astroparticle Physics and Space Science, Bose Institute, EN80, Bidhannagar, Kolkata 700091, India}

\author{Tomasz Paterek}
\affiliation{School of Physical and Mathematical Sciences, Nanyang Technological University, 637371 Singapore, Singapore}
\affiliation{MajuLab, CNRS-UCA-SU-NUS-NTU International Joint Research Unit, UMI 3654 Singapore, Singapore}

\date{\today}

\begin{abstract}
We characterise non-selective global projective measurements capable of increasing quantum entanglement between two particles. We show that non-selective global projective measurements are capable of increasing entanglement between two particles, in particular, entanglement of any pure non-maximally entangled state can be improved in this way (but not of any mixed state) and we provide detailed analysis for two qubits.
It is then shown that Markovian open system dynamics can only approximate such measurements,
but this approximation converges exponentially fast as illustrated using Araki-\.Zurek model.
We conclude with numerical evidence that macroscopic bodies in a random pure state do not gain entanglement in a random non-selective global measurement.
\end{abstract}

\maketitle

Quantum entanglement is a prerequisite for many quantum information applications and tests of foundations of physics~\cite{Horodecki_2009}.
It is typically created in a sequence of unitary processes (quantum gates) applied to an initially disentangled system, whose environment typically inhibits the non-classical correlations,
see e.g.~\cite{Julsgaard_2001,Steffen_2006,Blatt_2008,Palomaki_2013,Jabir_2017,Riedinger_2017}.
The environment does not always play the destructive role and another established route to entanglement generation is to engineer system-environment interactions in such a way that the decaying principal system reaches entangled steady state,
see e.g.~\cite{Krauter_2011,Verstraete_2009,Lin_2013,Shankar_2013,Reiter_2016}.
In this case, a manifold of initial system states evolve to the same entangled steady state.

Here we consider yet another possibility to generate or increase entanglement --- via global projective measurements with unknown results,
which may arise from the lack of ability to post-select a particular measurement result or even naturally since any measurement can be approximated with suitable open system dynamics.
Our approach is therefore related to engineering system-environment interactions but deviates from the decay processes because different initial states end up mostly in different final states having different amounts of entanglement.
It is also impossible to prepare maximally entangled states through any non-selective measurement, but dissipation allows for such a possibility~\cite{Lin_2013,Shankar_2013,Reiter_2016}.
(Actually, measurements along maximally entangled bases are useless for entanglement gain, modulo comments in Sec.\ref{SEC_MAX_ENT}).
Nevertheless, high entanglement can be observed as a result of non-selective measurement and, in fact, we show that entanglement of any pure non-maximally entangled state can be enhanced in this way, and we provide candidate optimal measurements for two qubits.
Next we show an example of Markovian open system dynamics which rapidly approximates non-selective measurement on the principal system.
In this way, we present another class of open system dynamics, additional to the decay processes, which increases entanglement in the system.
Finally, we study negativity gain with random initial states and random measurements for growing dimensionality of the subsystems.
We find that the probability of entanglement gain is negligible for higher dimensions.
This could be seen as an element of quantum-to-classical transition showing that macroscopic objects (of high dimensionality $d$) 
would have to be prepared in specific states and measured with specific non-selective measurements in order to entangle them.


\section{Black box problem}

\begin{figure}[b]
\includegraphics[scale=0.45]{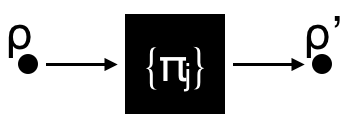}
\caption{Entanglement gain in non-selective global measurement as a black box problem.
A state $\rho$ is input to a black box measuring device which conducts von Neumann measurement along projectors $\{\Pi_j\}$.
Measurement results are unknown and therefore the state emerging from the box is $\rho' = \sum_j \Pi_j \rho \Pi_j$.
We study conditions under which entanglement in the state $\rho'$ is higher than entanglement of the input state $\rho$.}
\label{FIG_BB}
\end{figure}

We begin with the abstract formulation of the problem, described in Fig.~\ref{FIG_BB}.
A bipartite quantum system in a general mixed state $\rho$ enters a black box where projective measurement $\{\Pi_j\}$ is conducted but its results remain unknown.
That is, after the measurement, no single state $\Pi_j = | \psi_j \rangle \langle \psi_j |$ is selected for but rather we deal with a statistical mixture
\begin{equation}
	\label{eq:postmeasure}
	\rho' = \sum \limits_{j=0}^{D-1} \Pi_j \rho \Pi_j,
\end{equation}
where $D = d_1 d_2$ is the total dimension of tensor product space $\mathcal{H}_1 \otimes \mathcal{H}_2$
of Hilbert spaces $\mathcal{H}_1$ and $\mathcal{H}_2$ of dimensions $d_1$ and $d_2$ describing each subsystem respectively.
Throughout the paper, we assume without loss of generality that $d_1 \leq d_2$.
Our task is to study whether it is possible for quantum entanglement to increase as a result of the measurement, i.e. under what conditions the quantity
\begin{equation}
\Delta E \equiv E(\rho') - E(\rho)
\end{equation}
is positive for a suitable entanglement measure $E$, e.g. negativity~\cite{Zyczkowski_1998, Lee_2000a, Lee_2000b}.
\footnote{Note that if we allow general POVM measurements in place of projective von Neumann measurements the problem becomes trivial.
There is a POVM that simply ignores the input to the box and always outputs maximally entangled pure state $| \psi^{+} \rangle$.
For instance, take the POVM to be the collection of the following measurement operators: $M_j = | \psi^{+} \rangle \bra{\psi_j}$.}

\subsection{Upper bound on entanglement gain}

The following observation is an immediate consequence of the convexity of the entanglement measure $E$ and Eq.~\eqref{eq:postmeasure}.
	For any convex entanglement measure $E$,
	\begin{equation}
		E(\rho') \leq \max \limits_{j} E(\Pi_j).
		\label{EQ_MAX}
	\end{equation}
Note that inequality (\ref{EQ_MAX}) holds also for entanglement monotones which are not convex, but non-decreasing functions of convex ones.
A concrete example is the logarithmic negativity $LN(\rho) = \text{ln} (2 N(\rho) - 1)$,
where the negativity $N(\rho)$ of the state $\rho$ is the sum of moduli of all negative eigenvalues of the partial transposition $\rho^{\Gamma}$ of $\rho$~\cite{Zyczkowski_1998,Vidal_2002} and it is a convex function.

In other words, this observation says that there is no free lunch.
In order to end up with an entangled state at the output there must be a supply of entanglement in the basis states along which one measures.
This also suggests that it is desirable to choose the measurement basis states with as much entanglement as possible.
It turns out that this intuition is incorrect as the following section demonstrates.

\subsection{No gain with maximally entangled bases}
\label{SEC_MAX_ENT}

We now provide a proof that a measurement along any orthonormal basis composed of maximally entangled states of two qubits does not increase entanglement of any input state.
The idea is to use the fact that any such basis can be transformed as a whole, by local unitary operations, to the Bell basis.~\footnote{This is a folklore in the community, but since we were not able to locate a concrete reference we provide the complete proof as Lemma \ref{LEM_FOLKLORE} in the Appendix.}
We then note that the non-selective measurement along the Bell basis can be implemented by local operations and classical communication (LOCC).
Hence any measurement along maximally entangled basis has LOCC realisation and cannot increase entanglement.
We will comment on higher dimensions after presenting the proof.

\begin{theorem}
Assume that $D=4$ and let every state $\ket{\psi_j}$ in the measurement basis be maximally entangled.
Then for any input state $\rho$ and for any entanglement monotone $E$:
	$$E(\rho') \leq E(\rho),$$
where $\rho' = \sum_{j=0}^{3}  \ket{\psi_j} \bra{\psi_j} \rho \ket{\psi_j} \bra{\psi_j}$.
\end{theorem}
\begin{proof}
By Lemma \ref{LEM_FOLKLORE} (see Appendix),
the basis $\{ \ket{\psi_j} \}$ is locally unitarily equivalent to the Bell basis:
\begin{eqnarray}
| \psi_0 \rangle & = & \frac{1}{\sqrt{2}} (\ket{00} + \ket{11}), \quad
| \psi_1 \rangle = \frac{1}{\sqrt{2}} (\ket{01} + \ket{10}), \nonumber \\
| \psi_2 \rangle & = & \frac{1}{\sqrt{2}} (\ket{01} - \ket{10}), \quad
| \psi_3 \rangle = \frac{1}{\sqrt{2}} (\ket{00} - \ket{11}).
\end{eqnarray}
It is a matter of straightforward calculation to verify that
	\begin{eqnarray}
		\ket{\psi_0} \bra{\psi_0} & = &
		\frac{1}{4} \left(
			\identity \otimes \identity +
			\hat \sigma_1 \otimes \hat \sigma_1 -
			\hat \sigma_2 \otimes \hat \sigma_2 +
			\hat \sigma_3 \otimes \hat \sigma_3
		\right ), \nonumber \\
		\ket{\psi_1} \bra{\psi_1} & = &
		\frac{1}{4} \left(
			\identity \otimes \identity +
			\hat \sigma_1 \otimes \hat \sigma_1 +
			\hat \sigma_2 \otimes \hat \sigma_2 -
			\hat \sigma_3 \otimes \hat \sigma_3
		\right ), \nonumber \\
		\ket{\psi_2} \bra{\psi_2} & = &
		\frac{1}{4} \left(
			\identity \otimes \identity -
			\hat \sigma_1 \otimes \hat \sigma_1 -
			\hat \sigma_2 \otimes \hat \sigma_2 -
			\hat \sigma_3 \otimes \hat \sigma_3
		\right ), \nonumber \\
		\ket{\psi_3} \bra{\psi_3} & = &
		\frac{1}{4} \left(
			\identity \otimes \identity -
			\hat \sigma_1 \otimes \hat \sigma_1 +
			\hat \sigma_2 \otimes \hat \sigma_2 +
			\hat \sigma_3 \otimes \hat \sigma_3
		\right ), \nonumber
	\end{eqnarray}
where e.g. $\hat \sigma_1 = | 0 \rangle \langle 1| + | 1 \rangle \langle 0| $ stands for the $x$ Pauli operator.
Using these equations one finds after direct calculation the post-measurement state:
\begin{eqnarray}
\rho' & = & \frac{1}{4} (\rho + \hat \sigma_1 \otimes \hat \sigma_1 \, \rho \, \hat \sigma_1 \otimes \hat \sigma_1 \nonumber \\
& + & \hat \sigma_2 \otimes \hat \sigma_2 \, \rho \, \hat \sigma_2 \otimes \hat \sigma_2 +
		\hat \sigma_3 \otimes \hat \sigma_3 \, \rho \, \hat \sigma_3 \otimes \hat \sigma_3). \label{EQ_BELL_LOCC}
\end{eqnarray}
The same post-measurement state is obtained by the following LOCC procedure between the two observers, Alice and Bob.
Alice selects at random number $j$ and informs Bob about the value of this classical random variable.
Then Alice and Bob apply locally operation $\hat \sigma_j \otimes \hat \sigma_j$.
Finally, they erase the record of the values of $j$.
Therefore, for any entanglement monotone $E(\rho') \leq E(\rho)$.
\end{proof}

Unfortunately extension of this method to higher dimensions is not straightforward.
The properties of $\alpha^{(j)}$ matrices generalise and there is a similar LOCC procedure
that returns the post-measurement state which would be obtained by conducting non-selective measurement in the following straightforward generalisation of the Bell basis~\cite{teleportation}:
\begin{equation}
| \psi_{jk} \rangle = \frac{1}{\sqrt{d}} \sum_{m = 0}^{d-1} \omega_d^{m k} |m \rangle | m + j \rangle,
\label{EQ_MAX_END_D}
\end{equation}
where $j,k = 0,\dots, d - 1$, $\omega_d = \exp(i 2 \pi / d)$, the sum in the last ket is modulo $d$, and for simplicity we assume that both system have the same dimension $d$.
However, it is not clear that any maximally entangled basis can be transformed by local unitaries to the basis (\ref{EQ_MAX_END_D}).


\subsection{Pure input state}

Here we provide two theorems about possibility of entanglement gain if pure states are presented at the input.
In Theorem~\ref{TH_GIVEN_MEAS} we construct an input state whose entanglement increases in a given measurement.
It turns out that all non-trivial von Neumann measurement, with individual basis vectors being neither product nor maximally entangled states in some subspace, increase entanglement of suitable initial state.
In Theorem~\ref{TH_PURE_IN} we construct a measurement which improves entanglement of any pure input state which is not maximally entangled.

\begin{theorem}
\label{TH_GIVEN_MEAS}
Suppose that the measurement basis $\{ | \psi_j \rangle \}$ is such that $\ket{\psi_0}$ has at least two non-zero Schmidt coefficients that are not equal.
Then there exists a pure input state $\rho = \ket{\phi} \bra{\phi}$ for which negativity increases, i.e. $N(\rho') > N(\rho)$,
where $\rho' = \sum_{j = 0}^{D-1} | \langle \psi_j | \phi \rangle |^{2} \,
	\ket{\psi_j} \bra{\psi_j}$.
\end{theorem}
\begin{proof}
By assumption $\ket{\psi_0} = \sum_{i=0}^{d -1 } \sqrt{p_i} \ket{ii}$, where $d \leq d_1$, each $p_i > 0$, and $p_0 \neq p_1$.
We construct the input state $\ket{\phi}$ as follows:
\begin{equation}
\ket{\phi} \sim \ket{\psi_0} - \epsilon \ket{\Phi} = \sum_{i = 0}^{d-1} \left( \sqrt{p_i} - \epsilon / \sqrt{d} \right) | ii \rangle,
\label{EQ_PHI}
\end{equation}
where $\ket{\Phi} = \frac{1}{\sqrt{d}} \sum_{i=0}^{d-1} \ket{ii}$ is the maximally entangled state defined using the Schmidt basis of $\ket{\psi_0}$.
State $\ket{\phi}$ is normalised upon multiplication with constant $C^{-1} = (1 + \epsilon^2 - 2 \epsilon \beta)^{-1/2}$, where $\beta \equiv \langle \psi_0 | \Phi \rangle$.
Our aim is to derive the range of $\epsilon$ for which the post-measurement state has more negativity than the input state $\ket{\phi}$.
The post-measurement state reads:
\begin{equation}
\rho' = (1-p) \ket{\psi_0}\bra{\psi_0} + p \tilde{\rho},
\end{equation}
where $\tilde \rho$ lies in the orthogonal subspace to $\ket{\psi_0}$, and $p = 1 - | \langle \psi_0 | \phi \rangle |^2 = \frac{\epsilon^2}{C^2} (1 - \beta^2) > 0$ is the probability of obtaining the result different from zero.
By the triangle inequality of the trace norm:
\begin{eqnarray}
	|| \ket{\psi_0} \bra{\psi_0}^{\Gamma} ||_1 & \leq & \
	|| {\rho'}^\Gamma ||_1 +
	p || \tilde{\rho}^\Gamma - \ket{\psi_0}\bra{\psi_0}^\Gamma ||_1 \nonumber \\
	& \leq &
	|| {\rho'}^{\Gamma} ||_1 + 2 d p.
\end{eqnarray}
Since negativity of a generic state $\xi$ is given by $N(\xi) = (|| \xi^\Gamma||_1 - 1)/2$ we find that
\begin{equation}
N(\rho') \geq N(\ket{\psi_0}\bra{\psi_0}) - pd.
\label{EQ_FINAL_NEG}
\end{equation}
For the comparison with the initial state we calculate
\begin{equation}
N(\rho) = \frac{1}{C^2} \left[ N(\ket{\psi_0} \bra{\psi_0}) + \epsilon ( \epsilon - 2 \beta) N(\ket{\Phi}\bra{\Phi}) \right],
\label{EQ_NRHO_N0}
\end{equation}
where for simplicity we assumed that all the coefficients in Eq. (\ref{EQ_PHI}) are positive, i.e.
\begin{equation}
\epsilon < \sqrt{d} \, \min_i p_i.
\label{EQ_ESP_CONDITION}
\end{equation}
We now calculate $N(\ket{\psi_0}\bra{\psi_0})$ from (\ref{EQ_NRHO_N0}) and use it on the right-hand side of (\ref{EQ_FINAL_NEG}).
The requirement of negativity gain, i.e. $N(\rho') > N(\rho)$, holds for positive epsilon satisfying
\begin{equation}
\epsilon < \frac{2\beta [ N(\ket{\Phi}\bra{\Phi}) - N(\ket{\psi_1}\bra{\psi_1})]}
		{d(1-\beta^2) + N(\ket{\Phi}\bra{\Phi}) - N(\ket{\psi_1}\bra{\psi_1})}.
\end{equation}
We remind that $\epsilon$ also has to satisfy (\ref{EQ_ESP_CONDITION}).
\end{proof}

A simple example shows that the condition in Theorem \ref{TH_GIVEN_MEAS} is necessary.
Consider two qutrits and suppose we do a non-selective measurement along the following basis
\begin{eqnarray}
	\ket{\psi_0} &=& \frac{1}{\sqrt{2}} \left( \ket{00} + \ket{11} \right), \quad
	\ket{\psi_1} = \frac{1}{\sqrt{2}} \left( \ket{00} - \ket{11} \right),	\nonumber \\
	\ket{\psi_2} &=& \frac{1}{\sqrt{2}} \left( \ket{01} + \ket{10} \right),	 \quad
	\ket{\psi_3} = \frac{1}{\sqrt{2}} \left( \ket{01} - \ket{10} \right),	\nonumber \\
	\ket{\psi_4} &=& \ket{02}, \quad
	\ket{\psi_5} = \ket{12}, \quad
	\ket{\psi_6} = \ket{20},	\nonumber \\
	\ket{\psi_7} &=& \ket{21}, \quad
	\ket{\psi_8} = \ket{22}.	\label{EQ_LOCC_BASIS}
\end{eqnarray}
Since the non-selective measurement along this basis is implementable by LOCC (see Fig.~\ref{FIG_LOCC}), it follows that the entanglement can never increase on avarege.

\begin{figure}
	\includegraphics[width=.4\textwidth]{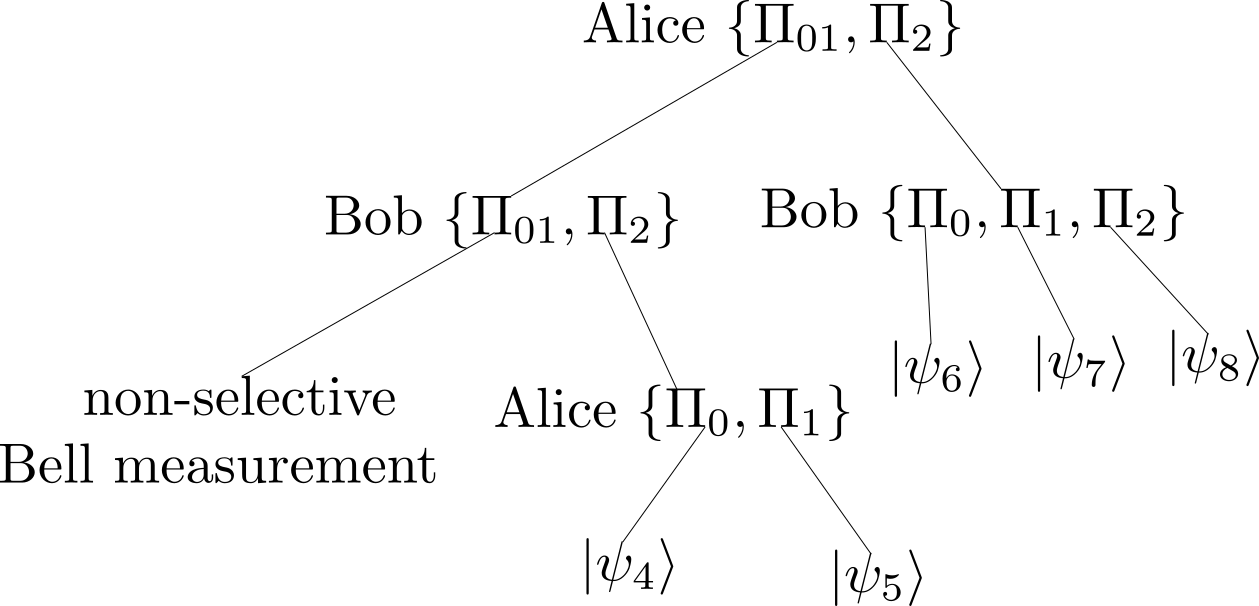}
	\caption{
	Graphical depiction of an LOCC protocol to implement non-selective measurement in basis (\ref{EQ_LOCC_BASIS}).
	Here $\Pi_{01}$ denotes $\proj{0} + \proj{1}$, and $\Pi_{i} = \proj{i}$.
	The LOCC procedure for the non-selective Bell measurement is described around Eq.~(\ref{EQ_BELL_LOCC}).
	}
	\label{FIG_LOCC}
\end{figure}

We have shown that for any non-trivial non-selective measurement there exists a state whose entanglement can be improved.
Now we show that given any input pure state, which is not maximally entangled, one can find a projective non-selective measurement that increases the amount of entanglement of the input state.

\begin{theorem}
\label{TH_PURE_IN}
	Suppose a pure initial state $\ket{\phi}$ is not maximally entangled.
	Then there exist a non-selective projective measurement $\{ \Pi_j \}$
	increasing negativity, i.e. $N(\sum_j \Pi_j | \phi \rangle \langle \phi | \Pi_j) > N(| \phi \rangle \langle \phi |)$.
\end{theorem}
\begin{proof}
Let us choose the measurement basis as follows:
\begin{eqnarray}
| \psi_0 \rangle & = & \frac{1}{\sqrt{2 + 2 \langle \phi | \Phi \rangle}} (\ket{\phi} + \ket{\Phi}), \nonumber \\
| \psi_1 \rangle & = & \frac{1}{\sqrt{2 - 2 \langle \phi | \Phi \rangle}} (\ket{\phi} - \ket{\Phi}), \label{EQ_MEAS_CHOICE}
\end{eqnarray}
where $\ket{\Phi} = \frac{1}{\sqrt{d_1}} \sum_{i=0}^{d_1-1} \ket{ii}$ is the maximally entangled state having the same Schmidt basis as $\ket{\phi}$.
One verifies that $| \psi_0 \rangle$ and $| \psi_1 \rangle$ are orthogonal as they should be and that $\ket{\phi}$ lies in the subspace they span, i.e. $\ket{\phi} = \alpha | \psi_0 \rangle + \beta | \psi_1 \rangle$.
The other basis states are arbitrary as they do not contribute to the post-measurement state
\begin{eqnarray}
\rho' & = & |\alpha|^2 | \psi_0 \rangle \langle \psi_0 | + |\beta|^2 | \psi_1 \rangle \langle \psi_1 | \nonumber \\
& = & \frac{1}{2} \ket{\phi} \bra{\phi} + \frac{1}{2} \ket{\Phi} \bra{\Phi},
\end{eqnarray}
where the second form of the state follows from (\ref{EQ_MEAS_CHOICE}).
Since both $\ket{\phi}$ and $\ket{\Phi}$ have the same Schmidt basis,
the final negativity is
\begin{eqnarray}
N(\rho') & = & \tfrac{1}{2} \left[ N(\ket{\phi} \bra{\phi}) + N(\ket{\Phi} \bra{\Phi}) \right] \\
& = & \tfrac{1}{2} \left[ N(\ket{\phi} \bra{\phi}) +  \tfrac{1}{2}(d_1 - 1) \right] > N(\ket{\phi} \bra{\phi}), \nonumber
\end{eqnarray}
where the value of negativity of the maximally entangled state was used.
\end{proof}


\subsection{Mixed input state}

A similar statement to Theorem~\ref{TH_PURE_IN}, but for all mixed input states, does not hold.
To see this, recall that a state $\rho$ is termed \emph{absolutely separable} \cite{Zyczkowski_1998},
if $U \rho U^{\dagger}$ is separable for every unitary $U$ acting on $\mathcal{H}_1 \otimes \mathcal{H}_2$.
We will argue that projective non-selective measurement can always be implemented as mixed unitary channel acting on the input.
Therefore, any absolutely separable input state leads to a separable post-measurement state.

\begin{proposition}
The measurement map $\rho \to \sum_j \Pi_j \rho \Pi_j$ can be implemented as $\rho \to \sum_j p_j U_j \rho U_j^\dagger$, i.e. as a mixture of unitary operations.
\end{proposition}
\begin{proof}
By construction. Consider unitary operator:
\begin{equation}
U \equiv \sum_{j=0}^{D-1} \omega_d^j \Pi_j.
\label{EQ_Z_DEF}
\end{equation}
The following probabilistic mixture mimics the non-selective measurement:
\begin{eqnarray}
\frac{1}{D} \sum_{i=0}^{D-1} U^i \rho (U^i)^\dagger & = & \frac{1}{D} \sum_{i,j,k=0}^{D-1} \omega_d^{i(j-k)} \Pi_j \rho \Pi_k \nonumber \\
& = & \sum_{j,k=0}^{D-1} \delta_{jk} \Pi_j \rho \Pi_k \nonumber \\
& = & \sum_{j = 0}^{D-1} \Pi_j \rho \Pi_j.
\end{eqnarray}
In the first line we used definition (\ref{EQ_Z_DEF}) and in the second line the following form of Kronecker delta: $\sum_{i=0}^{D-1} \omega_d^{i(j-k)} = D \delta_{jk}$.
\end{proof}

The class of absolutely separable states is rather non-trivial.
For example, it has been completely characterised for two qubits~\cite{Ishizaka_2000,Verstraete_2001}.
If $\lambda_1 \geq \lambda_2 \geq \lambda_3 \geq \lambda_4$ are eigenvalues of $\rho$, then $\rho$ is absolutely separable, if and only if $\lambda_3 + 2 \sqrt{\lambda_2 \lambda_4} - \lambda_1 \geq 0$.
A concrete not absolutely separable mixed state, whose entanglement increases as a result of non-selective measurement, is seen as follows.
Consider
$\rho = \frac{1+2\epsilon}{3} \ket{00}\bra{00} + \frac{1-\epsilon}{3} \ket{11}\bra{11} + \frac{1-\epsilon}{3} \ket{01}\bra{01}$,
for a small $\epsilon > 0$,
and the basis as in \eqref{EQ_OPT_BASIS} below, for $0<b<\frac{1}{\sqrt{2}}$ and $c=0$.
The state $\rho$ is separable and $\epsilon$-close to the set of absolutely separable states. 
It is easy to see, however, that the negativity of
the post-measurement state is strictly positive.
It is an intriguing question for further research whether all the states not from this class can gain entanglement in a suitable non-selective measurement.


\subsection{Two qubits}

We now provide results on the optimal measurement, giving the highest negativity gain, for any pure input state.
We begin with formulae valid for $d_1 = d_2 = d$ but soon restrict ourselves to $d = 2$, i.e. two qubits.
In this case we present proofs of optimality for a specific choice of measurement basis and end this section with numerical evidence that this choice of measurement basis is the best.

Consider a pure input state $\ket{\phi}$, endowed with Schmidt decomposition
\begin{equation}
\ket{\phi} = \sum_{i=0}^{d - 1} \sqrt{p_i} \ket{ii},
\end{equation}
where $p_i$ are probabilities.
We shall focus on the following specific choice of measurement basis.
We assume that the first $d$ vectors are spanned by the Schmidt basis of the input state, i.e.
\begin{equation}
\ket{\psi_j} = \sum_{i=0}^{d-1} \alpha_{ji} \ket{ii}, \quad j = 0,\dots,d-1,
\end{equation}
and all the remaining vectors $\ket{\psi_j}$, with $j \ge d$, are orthogonal to $\ket{\phi}$.
The state after the measurement reads
\begin{eqnarray}
	\rho' &=& \sum \limits_{j=0}^{d-1}  | \psi_j \rangle \langle \psi_j | \phi \rangle \langle \phi | \psi_j \rangle \langle \psi_j |  \\
	&=& \sum_{j, i, i'} \sum_{k, k'} \sqrt{p_i} \alpha_{ji} \sqrt{p_{i'}} \alpha_{ji'}^* \alpha_{jk} \alpha_{jk'}^* \ket{kk} \bra{k'k'}. \nonumber
\end{eqnarray}
Negativity of this state admits the lower bound
\begin{eqnarray}
	\label{eq:two_qudits_neg_derivation}
	N(\rho') &=& \sum_{k < k'} | \sum_{j, i, i'}
	\sqrt{p_i} \alpha_{ji} \sqrt{p_{i'}} \alpha_{ji'}^* \alpha_{jk} \alpha_{jk'}^* |	\\
	&\geq& | \sum_{k < k'} \sum_{j, i, i'} \sqrt{p_i} \alpha_{ji} \sqrt{p_{i'}} \alpha_{ji'}^* \alpha_{jk} \alpha_{jk'}^* |. \nonumber
\end{eqnarray}
It is therefore natural to define:
\begin{eqnarray}
	(\mathbf{m})_{ii'} &=& \sum_{k < k'}
	\sum_j \alpha_{ji} \alpha_{ji'}^* \alpha_{jk} \alpha_{jk'}^*	\\
	(\mathbf{q})_i		&=& \sqrt{p_i}, \nonumber
\end{eqnarray}
where matrix $\mathbf{m}$ depends only on the measurement,
and vector $\mathbf{q}$ depends only on the input state.
Then
\begin{equation}
	\label{eq:two_qudits_neg_inc}
	N(\rho') \geq | \langle \mathbf{q}, \mathbf{m} \mathbf{q} \rangle |.
\end{equation}

We now restrict ourselves to $d=2$.
A pure two-qubit input state in its Schmidt basis reads:
\begin{equation}
	\ket{\psi} = a \ket{00} + \sqrt{1 - a^2} \ket{11},
	\label{EQ_KETIN}
\end{equation}
where, without loss of generality, $a \in [0,1/\sqrt{2}]$.
Negativity of this input state is:
\begin{equation}
	N_i = a \sqrt{1 - a^2}.
	\label{EQ_NI}
\end{equation}
We take the following basis of the projective non-selective measurement:
\begin{eqnarray}
	\ket{\psi_1} & = & b \ket{00} + \sqrt{1 - b^2} \ket{11}, \label{EQ_OPT_BASIS} \\
	\ket{\psi_2} & = & \sqrt{1 - b^2} \ket{00} - b \ket{11}, \nonumber \\
	\ket{\psi_3} & = & c \ket{01} + \sqrt{1 - c^2} \ket{10}, \nonumber \\
	\ket{\psi_4} & = & \sqrt{1 - c^2} \ket{01} - c \ket{10}, \nonumber
\end{eqnarray}
where we note that the Schmidt basis of the input state is used to define every $\ket{\psi_j}$, but $\ket{\psi_3}$ and $\ket{\psi_4}$ do not contribute to the post-measurement state.
Using the previously introduced notation:
	\begin{equation}
		\mathbf{q} = ( a, \sqrt{1 - a^2} ),
	\end{equation}
	\begin{equation}
		\mathbf{m} =
		\begin{pmatrix}
			b               &  \sqrt{1 - b^2} \\
			\sqrt{1 - b^2}  &  -b
		\end{pmatrix}.
	\end{equation}
Hence, negativity of the relevant states of the measurement basis is fully characterised by parameter $b$ and we call the number $N_b = b \sqrt{1 - b^2}$ the ``negativity of the measurement''.
Again, without loss of generality, $b \in [0,1/\sqrt{2}]$.
We observe that in the case of two qubits the outer sum in \eqref{eq:two_qudits_neg_derivation}
contains just one term and the inequality \eqref{eq:two_qudits_neg_inc} becomes an equality.
The negativity of the final state $\rho'$, expressed in terms of $N_i$ and $N_b$, is:
\begin{equation}
	N_f \equiv N(\rho') =  N_b \sqrt{1 - 4 N_i^2} \sqrt{1 - 4 N_b^2} + 4 N_i N_b^2.
\end{equation}
The condition for entanglement gain, $N_f > N_i$, is realised for measurements and states that satisfy $N_i < N_b < \frac{1}{2}$,
confirming what we have we already said about the lack of free lunch and the impossibility of entanglement enhancement in a measurement with maximally entangled basis states.

The best measurement for a given initial state (i.e. producing the largest final negativity) is obtained by solving $\frac{\partial N_f}{\partial N_b} = 0$ for $N_b$.
The corresponding negativity of the measurement reads
\begin{equation}
	N_b^{\max} = \frac{\sqrt{2 N_i +1}}{2 \sqrt{2}},
\end{equation}
and leads to entanglement gain decreasing linearly with the initial negativity:
\begin{equation}
	N_f^{\max} - N_i = \frac{1}{4} - \frac{N_i}{2}.
	\label{EQ_N_INCREMENT}
\end{equation}
Therefore, the largest negativity gain is for initial product state measured along a basis having negativity
$N_b^{\max} = 1/2 \sqrt{2}$.
This also shows again that any initial pure state has a strictly positive negativity gain in this process,
unless it is maximally entangled.

\begin{figure}[b]
\includegraphics[scale=0.55]{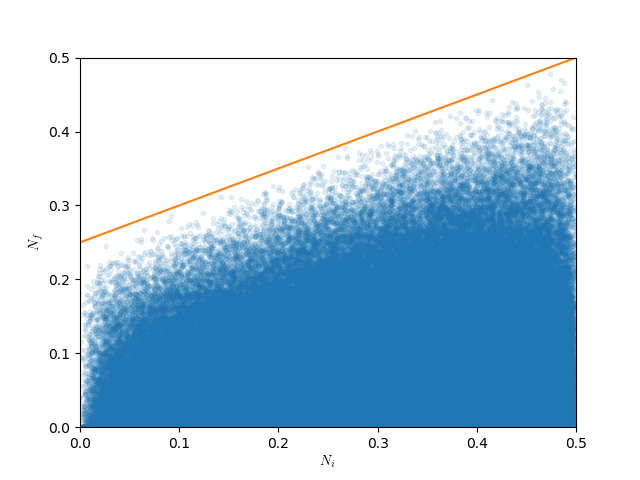}
\caption{Random sampling supports our claim that for the initial state (\ref{EQ_KETIN}) the measurement giving rise to highest entanglement gain is given by (\ref{EQ_OPT_BASIS}).
Each point represents negativity gain for a random input state measured along random basis.
Both states and measurements are sampled uniformly at random according to the Haar measures.
One million trials are recorded.
The upper boundary is given by Eq.~(\ref{EQ_N_INCREMENT}) derived for the candidate optimal measurement, i.e. given by (\ref{EQ_OPT_BASIS}).}
\label{FIG_OPT}
\end{figure}

Finally, we address the issue whether results just presented are globally optimal.
In order to derive them we assumed the measurement basis of the form (\ref{EQ_OPT_BASIS}).
Although we were not able to provide a proof that this is indeed the optimal choice among all possible measurement bases,
we have conducted numerical check which supports optimality.
Fig.~\ref{FIG_OPT} presents the data from a million of random trials, where each data point corresponds to a pair of random pure input state and random measurement.
We plot negativity gain $N_f - N_i$ as a function of the negativity of the initial state, $N_i$.
We observe linear boundary matching Eq.~(\ref{EQ_N_INCREMENT}) derived for the measurement in (\ref{EQ_OPT_BASIS}).


\section{Projective measurement in open system dynamics}

We have shown that non-selective global projective measurements can lead to entanglement gain in a bipartite system.
Here we ask if such measurements could emerge as a result of open system dynamics.
We demonstrate that a Markovian time evolution cannot lead to a perfect projective measurement in finite time.
However, it rapidly converges to the projective measurement as illustrated using Araki-\.Zurek model~\cite{Araki_1980,Zurek_1982}.

Suppose that $\mathcal{E}_t: \text{M}_{n} \rightarrow \text{M}_{n}$
is a Markovian evolution on the space of matrices of size $n$,
which is norm-continuous and $\mathcal{E}_0 = \identity$.
We require that this evolution reaches a projective map at least asymptotically , i.e. $\lim_{t \rightarrow \infty} \mathcal{E}_t = \Pi$,
where $\Pi$ is a map with rank one projectors $\{ \Pi_j \}$ as in our black-box problem.
It turns out that $\mathcal{E}_t$ cannot reach the projective map at any earlier time.
The argument below shows that if this was the case the evolution would have to be a projector even for tiny evolution times which contradicts continuity,
i.e. in this limit evolution map has to approach identity (nothing happens in very short times).

\begin{theorem}
	$\mathcal{E}_{\tau} \neq \Pi$ for any finite time $\tau$.
\end{theorem}
\begin{proof}
By contradiction.
Suppose there is finite $\tau$ such that $\mathcal{E}_{\tau} = \Pi$.
Let $\Pi' = \openone - \Pi$, where $\openone$ is the identity map, i.e. the action of $\Pi'$ on a density matrix is to erase the diagonal elements in the basis defined by $\{ \Pi_j \}$.
We also define $\mathcal{E}_{t}' = \Pi' \circ \mathcal{E}_{t}$ and therefore $\mathcal{E}_{\tau}' = 0$.
Our aim is to show that $\mathcal{E}_{t}' = 0$ also for all $t < \tau$, i.e. $\mathcal{E}_{t}$ is a projector for all evolution times.

To this end we recall the definition of Markovianity $\mathcal{E}_{t + s} = \mathcal{E}_{t} \circ \mathcal{E}_{s} = \mathcal{E}_{s} \circ \mathcal{E}_{t}$.
Using it one verifies similar property for the primed maps:
\begin{equation}
\mathcal{E}_{t + s}' = \mathcal{E}_{t}' \circ \mathcal{E}_{s}' = \mathcal{E}_{s}' \circ \mathcal{E}_{t}'.
\label{EQ_PRIME_COMM}
\end{equation}
Note that since $\mathcal{E}_{t}'$ is a linear transformation, it can be represented as a $n^2 \times n^2$ matrix.
Let us now choose an integer $m$ and study matrix representation $A$ of the map $\mathcal{E}_{\tau / m}'$ (we put no limit on how high $m$ could be).
By our assumption and by property (\ref{EQ_PRIME_COMM}) we conclude that $A^m = 0$.
This implies that all the eigenvalues of matrix $A$ are zero, and therefore it has a simple Jordan normal form
\begin{equation}
A = S \left( \begin{array}{cccc}
0 & 1 & & \\
& 0 & 1 & \\
& & \ddots & \\
& & & 0
\end{array}
\right) S^{-1},
\end{equation}
with $1$s on suitable positions right above the diagonal and with $S$ being an invertible matrix.
Since matrix $A$ has size $n^2 \times n^2$, there exists an integer $\mu < n^2$ for which $A^\mu = 0$ because multiplying A with itself shifts the 1s further from the diagonal.
Since $\mu$ is independent of $m$ the map $\mathcal{E}_{\frac{\mu}{m} \tau}' = 0$ for any rational fraction of $\tau$.
For any other $t$ we note using property (\ref{EQ_PRIME_COMM}) again that $\mathcal{E}_{t}' = \mathcal{E}_{t - \frac{\mu}{m} \tau}' \circ \mathcal{E}_{\frac{\mu}{m} \tau}' = 0$.
We have shown that $\mathcal{E}_{t} = \Pi$ for any $0 < t < \tau$.
Since the map $t \mapsto \mathcal{E}_{t}$ is continuous, it is required that $\mathcal{E}_{0} = \Pi$, which amounts to contradiction.
\end{proof}

To illustrate how open quantum dynamics can lead to a projective measurement,
let us take concrete example within the Araki-\.Zurek model \cite{Araki_1980,Zurek_1982}.
Consider the system of two qubits coupled to the environment represented by a free particle moving on the real line.
We wish to demonstrate entanglement gain for the initial state $\rho_0 = \ket{00}\bra{00}$
and therefore choose the measurement basis:
\begin{eqnarray}
	\ket{\psi_0} & = & b \ket{00} + \sqrt{1 - b^2} \ket{11}, \\
	\ket{\psi_1} & = & \sqrt{1 - b^2} \ket{00} - b \ket{11}, \nonumber \\
	\ket{\psi_2} & = & \ket{01}, \quad
	\ket{\psi_3} = \ket{10}, \nonumber
\end{eqnarray}
with $0 < b < \frac{1}{\sqrt{2}}$.
Recall, that the post-measurement state reads $\rho' = \sum_{j} \langle \psi_j | \rho | \psi_j \rangle \, \ket{\psi_j} \bra{\psi_j}$.

Let us now define the total Hamiltonian for the evolution of the system and the environment.
We put
	\begin{equation}
		H = H_S \otimes \identity_E + \identity_S \otimes H_E + A \otimes B,
	\end{equation}
where $H_S = \sigma_z \otimes \identity - \identity \otimes \sigma_z$ is the hamiltonian for the system,
$H_E = \frac{\hat{p}^2}{2 m}$ is the hamiltonian of the environment
and where $\hat{p}$ is the momentum operator of the particle with mass $m$.
The interaction term is given by $A = \sum_{j} \lambda_j \ket{\psi_j}\bra{\psi_j}$,
with $\lambda_j \neq \lambda_{j'}$ for $j \neq j'$, together with $B = \hat{p}$.
Observe that $[H_S , A ] = 0$ as well as $[H_E, B] = 0$.
The initial state of the environment is chosen to be
$\omega_E = \ket{\phi_E} \bra{\phi_E}$,
where $\phi_E(x) = \frac{1}{\sqrt{2 \pi}} \int \frac{e^{i x p}}{\sqrt{\pi (1 + p^2)}} \, dp$ is the Cauchy distribution in momentum space here ensuring that the evolution is Markovian~\cite{Blanchard_2003}.

Starting from the uncorrelated initial state $\rho_0 \otimes \omega_E$ the time evolution of the system alone is given by the Eq. (7) in \cite{Blanchard_2003}:
	\begin{equation}
		\rho_t = \sum \limits_{n,m = 0}^3
		e^{- | \lambda_n - \lambda_m| \, t}
		e^{it (\gamma_n - \gamma_m)}
		\langle \psi_n | \rho_0 | \psi_m \rangle
		\ket{\psi_n} \bra{\psi_m},
	\end{equation}
where $\gamma_0 = \gamma_1 = 0$ and $\gamma_2 = - \gamma_3 = 2$.
Thus
\begin{eqnarray}
\rho_t & = & b^2 \ket{\psi_0} \bra{\psi_0} +
		(1-b^2) \ket{\psi_1} \bra{\psi_1} \nonumber \\
& + &
		e^{-|\mu| t} b \sqrt{1 - b^2}
		\left( \ket{\psi_0} \bra{\psi_1} + \ket{\psi_1} \bra{\psi_0} \right),
\end{eqnarray}
where $\mu = \lambda_1 - \lambda_2 \neq 0$.
The negativity of $\rho_t$ is
\begin{equation}
N(\rho_t) = b \sqrt{1-b^2} |2b^2 - 1| (1 - e^{-|\mu|t}).
\end{equation}
It is easy to see that in the limit of $t \rightarrow \infty$ the system state approaches the post-measurement state $\rho_t \rightarrow \rho'$ in the trace norm.
Note also that this convergence is exponentially fast, scaling as $\exp{(-|\mu|t)}$.


\section{Random states and random measurements}

Here we provide several results about entanglement gain with random measurements and random pure initial states on two $d$-level systems.\\
We found numerically that the probability of negativity gain is very small for $d = 2$ ($\approx 1.6\%$), whereas for higher $d$ it is smaller than machine precision of present day computing clusters.
In order to understand this, we study the distribution of negativity in a random pure input state of two qudits
and compare it to the distribution of negativity in the post-measurement state.

\begin{table}
\begin{tabular}{c|c}
	dim	&	$\langle \mathcal{S} \rangle$	\\ 
	\hline
	2	&	0.063 \\ 
	3	&	0.054 \\ 
	4	&	0.045 \\ 
\end{tabular}
\caption{Estimated mean of the statistical distance (\ref{EQ_STAT_DIST}) obtained by sampling $10^4$ input states.
It shows that the probability distribution of negativity in the post-measurement state is practically independent of the input state.}
\label{TAB_NEG_DIST}
\end{table}

First we argue that the negativity distribution in the post-measurement state is practically independent of the input state.
Consider an input pure state $\ket{\phi}$ measured in a uniformly random basis.
We denote by $P(N_f | \phi)$ the probability of observing an output state with negativity $N_f$.
We have sampled from this distribution $10^4$ times for a given input state and numerically approximated it by dividing the range of negativity (note that it depends on $d$) into $10^3$ bins.
In this way we discretise possible values of negativity justifying out notation $P(N_f | \phi)$ instead of probability density suitable for continuous variables.
We then additionally sampled uniformly at random the input states, $10^4$ times, in order to approximate $P(N_f)$, probability distribution of negativity in the output state averaged over the input states.
Tab.~\ref{TAB_NEG_DIST} shows that the two distributions are very similar as revealed by the estimated mean value of the statistical distance between them
\begin{equation}
\mathcal{S}(\phi) = \sum_{N_f} |P(N_f | \phi) - P(N_f)|. 
\label{EQ_STAT_DIST}
\end{equation}

We plot the distribution of negativity in the post-measurement state in Fig.~\ref{FIG_RANDOM} for two qudits with $d = 2,3,4$.
For comparison we also present the negativity distribution in a randomly chosen input state.
As clearly seen only for $d=2$ the two distributions have non-negligible overlap.
It is also clear that the chances of increasing negativity via a random measurement are negligible for all $d > 3$.

\begin{figure}
	\includegraphics[scale=0.55]{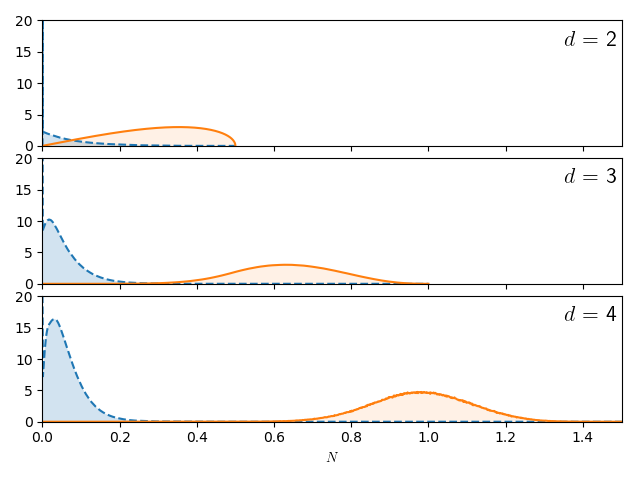}
	\caption{Distribution of negativity between two $d$-level systems.
	Dashed lines enclosing blue colour show the distributions for the post-measurement state
	whereas solid lines enclosing orange colour show the distribution for a random pure input states.}
	\label{FIG_RANDOM}
\end{figure}


\section{Conclusions}

We proposed the use of global non-selective projective measurements as means for entanglement creation or increase.
Excellent approximations to such measurements may occur naturally in open system dynamics or they may also appear in engineered systems due to the lack of ability to post-select a particular measurement result.
We showed that entanglement of any pure non-maximally entangled state can be improved in this way, but the final state can never be maximally entangled.
In fact, measurements along bases with solely maximally entangled states provide no entanglement gain whatsoever (up to comments in Sec.~\ref{SEC_MAX_ENT}).
We derived candidate optimal measurements, giving maximal entanglement gain, for any pure input state of two qubits.
Numerical checks support the optimality.
It turns out that it is best to start with a product state and measure it along the basis with states having moderate negativity of $1 / 2 \sqrt{2}$.
We hope these results on entanglement gain will find concrete realisations leading to new ways of generating this important physical resource.

\acknowledgments

We thank Micha{\l} Horodecki, Kavan Modi, Aby Philip and Somasundaram Sankaranarayanan for discussions.
SB thanks Anindita Banerjee and Saronath Halder for helpful discussions.
This work is supported by Singapore Ministry of Education Academic Research Fund Tier 2 Project No. MOE2015-T2-2-034.
SB is supported in part by SERB project EMR/2015/002373.

\section{Appendix}

\begin{lemma}
	\label{LEM_FOLKLORE}
Let $d_1 = d_2 = 2$. Assume that every state $\ket{\psi_j}$ of an orthonormal basis is maximally entangled.
Then the basis is locally unitarily equivalent to the Bell basis.
\end{lemma}
\begin{proof}
Let $\{ \ket{a_k b_l} \}$ be a product-vector basis, for which $\ket{\psi_0} = \frac{1}{\sqrt{2}}(\ket{00} + \ket{11})$.
There always exists such a basis due to Schmidt decomposition.
We expand the rest of the states $\ket{\psi_j}$ in the same basis:
\begin{equation}
\ket{\psi_j} = \sum \limits_{k,l=0}^{1} \alpha^{(j)}_{kl} \ket{a_k b_l},
\end{equation}
and consider $2 \times 2$ matrices $\alpha^{(j)}$ with entries given by the coefficients of the state $\ket{\psi_j}$.
They have the following properties stemming from maximal entanglement and orthonormality of these states:
\begin{eqnarray}
\textrm{max. entanglement} & \iff & \alpha^{(j)} \textrm{ unitary} \nonumber \\
\langle \psi_j | \psi_{j'} \rangle = \delta_{j j'} & \iff & \Tr[\alpha^{(j)} ( \alpha^{(j')} )^\dagger] = \delta_{j j'}
\end{eqnarray}
Note that $\Tr(\alpha^{(j)}) = 0$ for $j = 1,2,3$.
Any traceless $2 \times 2$ unitary matrix is proportional to a hermitian one, i.e. each $\alpha^{(j)} = e^{i \eta_j} A^{(j)}$, where $A^{(j)}$ is Hermitian.
Since multiplying matrices $\alpha^{(j)}$ by a phase factor $e^{i \eta_j}$ does not change the post-measurement state $\rho'$, we can assume without loss of generality that each $\alpha^{(j)}$ itself is Hermitian, for $j = 1,2,3$.
The set $\{ \frac{1}{\sqrt{2}} \identity, \alpha^{(1)}, \alpha^{(2)}, \alpha^{(3)} \}$ is therefore an orthonormal basis in the four-dimensional real Hilbert space of hermitian $2 \times 2$ matrices with the Hilbert-Schmidt inner product.
Since the same space is spanned by the identity and the Pauli matrices there is an isometry $R$ such that
\begin{eqnarray}
\identity & = & R \identity \nonumber \\
\alpha^{(j)} & = & R ( \tfrac{1}{\sqrt{2}} \sigma_j) \quad \textrm{for } j = 1,2,3,
\end{eqnarray}
where $\sigma_j$ is the $j$th Pauli matrix.
The matrix representation of the mapping $R$ reads
	\begin{equation}
		R = \begin{pmatrix}
			1 & 0 \\
			0 & R_0
		\end{pmatrix},
	\end{equation}
	where $R_0 \in \text{O}(3)$ is an orthogonal matrix.
	We can always assume that $R_0 \in \text{SO}(3)$.
	Indeed, if $\text{det} R_0 = -1$, we would take
	one of the measurement basis vectors,
	say $\ket{\psi_1}$,
	with the opposite sign, which amounts to a global phase change.
	That would change the sign of $\alpha^{(1)}$
	and consequently the sign of $\text{det} R_0$,
	but the post-measurement state $\rho'$ would stay the same.

Using the homomorphism between the groups $\text{SU}(2)$ and $\text{SO}(3)$ we infer that there exists a unitary matrix $U \in \text{SU}(2)$, such that
$\alpha^{(j)} =  \frac{1}{\sqrt{2}} R_0 \sigma_j = \frac{1}{\sqrt{2}} U \sigma_j U^{\dagger}$, for $j = 1,2,3$, and therefore:
\begin{equation}
| \psi_j \rangle = \frac{1}{\sqrt{2}} \sum_{k,l =0}^1 (U \sigma_j U^\dagger)_{kl} |a_k b_l \rangle.
\end{equation}
Note that last equation holds for all $j$ because $\alpha^{(0)} = \identity$.
By opening up matrix multiplication $U \sigma_j U^\dagger$ and defining new local bases
\begin{eqnarray}
\ket{m} & \equiv & \sum_k U_{km} \ket{a_k}, \nonumber \\
\ket{n} & \equiv & \sum_l U_{ln}^* \ket{b_l},
\end{eqnarray}
the initial maximally entangled basis is brought to the form
\begin{equation}
| \psi_j \rangle = \frac{1}{\sqrt{2}} \sum_{m,n=0}^1 (\sigma_j)_{mn} \ket{mn}.
\end{equation}
This is, up to irrelevant global phase of $\ket{\psi_2}$, the standard Bell basis.
\end{proof}

\bibliography{refs}

\end{document}